\documentclass[11pt,a4paper]{article}
\usepackage[utf8]{inputenc}
\usepackage[margin=1in]{geometry}
\usepackage{amsmath,amssymb,amsthm}
\usepackage{hyperref}
\usepackage{natbib}
\usepackage{booktabs}
\usepackage{graphicx}
\usepackage{tikz}
\usetikzlibrary{arrows.meta, positioning}

\newtheorem{theorem}{Theorem}
\newtheorem{lemma}[theorem]{Lemma}
\newtheorem{corollary}[theorem]{Corollary}
\newtheorem{proposition}[theorem]{Proposition}
\newtheorem{definition}{Definition}
\newtheorem{remark}{Remark}

\newtheorem{assumption}{Assumption}

\title{On the Inseparability of Instructions and Data\\in Shared-Embedding Sequence Models}

\author{Dewank Pant\\
Independent Researcher\\
\texttt{dewankpant@gmail.com}
\and
Shruti Lohani\\
Independent Researcher\\
\texttt{shrutilohani9@gmail.com}
\and
Avijit Kumar\\
Independent Researcher\\
\texttt{avijitkumar2002@gmail.com}}
\date{June 2026}

\begin{document}
\maketitle

\begin{abstract}
Prompt injection is widely treated as a defect to be patched through better training,
filtering, or guardrails. We show that, for a broad class of prompted neural architectures,
perfect prompt-injection prevention cannot be guaranteed within the shared representational
pipeline itself. We formalize prompted systems as \emph{Prompted Action Models} whose outputs
include not only token generation but also control-authoritative actions such as refusal decisions, tool authorization, policy routing, and memory-write operations. We then define
a security property, \emph{Semantic-Faithful Control} (SFC), requiring that control-authoritative
behavior depend only on the semantic content of untrusted input, not on semantically irrelevant variations in its representational encoding.

Our analysis has three parts. First, we derive a provenance-recovery impossibility result:
when trusted instructions and untrusted content are processed through shared representations, the Bayes-optimal error of provenance recovery is governed by the total variation
distance between the corresponding representation distributions, and is nonzero whenever
those distributions are not disjoint. Second, we show that in standard shared-attention architectures, untrusted content enters control-relevant computation through the same value-aggregation pathway used to determine control-authoritative outputs. Third, combining
these facts with a finite-coverage invariance argument, we prove that no mechanism operating
solely within the shared representational pipeline can guarantee perfect Semantic-Faithful
Control.

Our result does not imply that useful systems must ignore user input, nor that all defenses are equally ineffective in practice. Rather, it establishes a sharp limit: within shared-embedding architectures lacking enforced separation between trusted control and untrusted
content, perfect prompt-injection prevention is impossible. This result is structurally analogous to the code-data confusion in Von Neumann machines that gives rise to buffer overflow
vulnerabilities. That problem required a decades-long layered defense combining architectural mitigations (DEP, W$\oplus$X), runtime protections (ASLR, stack canaries), and ultimately
memory-safe languages, because no single mechanism was sufficient. The implication is architectural: complete protection requires mechanisms that enforce control separation, rather
than solely improving in-pipeline classification or alignment.
\end{abstract}


\section{Introduction}
\label{sec:intro}

The integration of Large Language Models (LLMs) into applications that process untrusted user
input has created a class of vulnerabilities collectively known as \textit{prompt
injection}~\citep{greshake2023indirect, liu2024formalizing}. In a prompt injection attack, an
adversary crafts user data that, when processed by the model alongside system instructions,
causes the model to deviate from its intended behavior, executing the attacker's instructions
instead of (or in addition to) the system's.

The research community has produced a proliferation of defenses: instruction tuning, RLHF,
guardrail classifiers, perplexity filtering, sandwich defenses, XML tagging, representation
engineering, and circuit breakers~\citep{zou2024circuit}. Yet no proposed defense has proven
robust; many are broken within weeks of publication. Multi-turn attacks that decompose harmful
queries across conversation turns reach success rates of 90--99\% on open-weight
models~\citep{jailbreak2026survey}. Even circuit breakers, the current state-of-the-art
representation-level defense, fall to 54.2\% ASR under Crescendo
attacks~\citep{crescendo2025repe}. Universal adversarial suffixes~\citep{zou2023universal}
demonstrate that semantically negligible token strings appended to harmful requests can
systematically flip refusal behavior across multiple aligned models.

This paper asks a more fundamental question: \textbf{Is perfect prompt injection prevention
even theoretically possible within the current architecture?}

We prove the answer is \textbf{no}. Specifically, we prove that in shared-embedding architectures
lacking enforced control-data separation, the shared representational space in which transformers
process both instructions and data makes perfect prevention of prompt injection mathematically
impossible. This is not a statement about the inadequacy of current defenses. It is a statement
about the architecture itself.

\paragraph{Historical Precedent: The Von Neumann Analogy.} Our result has a precise historical
parallel. The Von Neumann architecture stores code and data in shared memory. This design
choice, elegant and powerful, is also the root cause of buffer overflow vulnerabilities: an
attacker can craft data that, when written to memory, is \textit{interpreted as code}. This
vulnerability class was discovered in the 1970s~\citep{anderson1972planning} and required
roughly three decades of layered defenses to meaningfully contain: source-level bounds checking
and safer library functions, stack canaries (StackGuard, 1998), Address Space Layout
Randomization (ASLR), the hardware No-Execute bit (NX/DEP, AMD 2003) enforcing
Write-XOR-Execute policies, control-flow integrity, and ultimately the adoption of memory-safe
languages (Rust, Go, Java) that eliminate the bug class at the language level. Each mechanism
closed specific exploitation patterns, but the full defense posture required architectural,
runtime, and language-level interventions in combination.

We prove that the transformer architecture has an analogous structural vulnerability. The
embedding function maps both instructions and data into the same vector space. Once embedded,
their types are erased. Just as Von Neumann machines cannot distinguish code from data in
shared memory without external enforcement, transformers cannot guarantee that
control-authoritative behavior depends only on semantic content rather than on adversarial
representational manipulation, absent architectural separation.

\paragraph{Contributions.} We make the following contributions:
\begin{enumerate}
    \item We formalize \textit{Semantic-Faithful Control} (SFC), a security property requiring
    that control-authoritative behavior depend on the semantic content of untrusted input, not on
    semantically irrelevant variations in encoding. We show that SFC is the minimal security
    property for agentic systems with tool use, memory, and policy routing.

    \item We prove the \textit{Inseparability Theorem}: in any shared-embedding architecture
    with exposed control paths and without immutable provenance enforcement, perfect SFC cannot
    be guaranteed within the shared representational pipeline. The proof combines a
    provenance-recovery impossibility (via total variation distance), a control-path exposure
    result (untrusted values enter control computation via attention), and a finite-coverage
    invariance argument (finite training cannot certify invariance over infinite encoding
    equivalence classes).

    \item We formalize the structural isomorphism between prompt injection and buffer overflows,
    grounding the analogy in a precise correspondence between architectural properties.
\end{enumerate}

\section{Formal Framework}
\label{sec:framework}

\subsection{The Prompted Action Model}

\begin{definition}[Prompted Action Model]
\label{def:pam}
A \textbf{Prompted Action Model} is a tuple
$\mathcal{M} = (\Sigma, d, E, f, \mathcal{A})$
where:
\begin{itemize}
    \item $\Sigma$ is a finite token vocabulary with $|\Sigma| = V$.
    \item $d \in \mathbb{N}$ is the embedding dimension.
    \item $E : \Sigma \to \mathbb{R}^d$ is a \textbf{shared embedding function}, extended to
    sequences as $E : \Sigma^* \to (\mathbb{R}^d)^*$.
    \item $\mathcal{A}$ is a finite or countable \textbf{action space} comprising all externally
    observable model actions, including token generation, tool invocation, policy routing, and
    memory-write operations.
    \item $f : (\mathbb{R}^d)^* \to \Delta(\mathcal{A})$ is the model function mapping a
    sequence of embeddings to a probability distribution over actions in $\mathcal{A}$.
\end{itemize}
\end{definition}

\begin{definition}[Prompted Input with Provenance]
\label{def:prompted-input}
A \textbf{prompted input} is a concatenation $x = (s \| u)$ where:
\begin{itemize}
    \item $s \in \Sigma^*$ is the \textbf{trusted instruction prefix} (specified by the application
    developer).
    \item $u \in \Sigma^*$ is the \textbf{untrusted user content} (specified by the user or
    adversary).
    \item $\|$ denotes concatenation in the token sequence.
\end{itemize}
Each position in $x$ carries a provenance variable $R_j \in
\{\texttt{trusted}, \texttt{untrusted}\}$. The model processes the combined embedding
sequence $E(s \| u) = (E(x_1), \ldots, E(x_{|x|}))$.
\end{definition}

\begin{remark}[The Shared Embedding Axiom]
\label{rem:shared}
The critical architectural property: both $s$ and $u$ are mapped through the \textbf{same}
embedding function $E$. After embedding, there is no type tag, metadata, or immutable
architectural feature that distinguishes a trusted embedding from an untrusted embedding. They are
vectors in the same space $\mathbb{R}^d$.
\end{remark}

\subsection{The Security Property: Semantic-Faithful Control}

\begin{definition}[Semantic-Faithful Control]
\label{def:sfc}
Let $Y(x) \in \Delta(\mathcal{A})$ denote the model's action distribution on prompted input
$x = (s \| u)$, and let
\[
\Pi_{\mathrm{ctrl}} : \Delta(\mathcal{A}) \to \Delta(\mathcal{A}_{\mathrm{ctrl}})
\]
denote projection onto the \textbf{control-authoritative outputs}, where
$\mathcal{A}_{\mathrm{ctrl}} \subseteq \mathcal{A}$ comprises refusal decisions,
tool-authorization decisions, policy-routing decisions, and memory-write authorization decisions.

Let $\equiv_{\mathrm{sem}}$ be a \textbf{semantic equivalence relation} on untrusted inputs,
where $u \equiv_{\mathrm{sem}} u'$ means that $u$ and $u'$ express the same task-relevant
semantic content for the application under consideration.

A prompted action model satisfies \textbf{Semantic-Faithful Control (SFC)} if, for all trusted
prefixes $s \in \Sigma^*$ and all untrusted inputs $u, u' \in \Sigma^*$,
\begin{equation}
    u \equiv_{\mathrm{sem}} u'
    \quad \Longrightarrow \quad
    \Pi_{\mathrm{ctrl}} Y(s \| u) = \Pi_{\mathrm{ctrl}} Y(s \| u').
    \label{eq:sfc}
\end{equation}
That is, control-authoritative behavior may depend on the semantic content of the untrusted
input (e.g., refusing harmful requests), but not on semantically irrelevant variations in its
representation or encoding.
\end{definition}

\begin{remark}[Parametricity of $\equiv_{\mathrm{sem}}$]
\label{rem:sem-parametric}
The theorem is parametric in the choice of $\equiv_{\mathrm{sem}}$. For any nontrivial semantic
equivalence (one that identifies at least two distinct token sequences as equivalent), the
impossibility holds. We instantiate $\equiv_{\mathrm{sem}}$ conservatively as \textit{encoding
equivalence}: $u \equiv_{\mathrm{sem}} u'$ if $u$ and $u'$ differ only in
semantics-preserving encoding transformations such as character-level obfuscation, Unicode
homoglyph substitution, base64 encoding, or the addition of semantically vacuous token strings.
Any broader equivalence only strengthens the result.
\end{remark}

\begin{remark}[Constructive SFC Violation via Universal Adversarial Suffixes]
\label{rem:gcg-sfc}
The universal adversarial suffixes of \citet{zou2023universal} provide a concrete, published
demonstration of an SFC violation. Appending a GCG-optimized suffix does not add task-relevant
semantic content to the underlying request under any reasonable $\equiv_{\mathrm{sem}}$, yet it
can flip the projected control-authoritative behavior $\Pi_{\mathrm{ctrl}} Y$ from refusal to
compliance on real aligned models. Thus, GCG constitutes a constructive violation of
Semantic-Faithful Control: semantically equivalent inputs induce different control-authoritative
outputs solely through representational manipulation within the shared pipeline.
\end{remark}

\begin{remark}[What SFC Does and Does Not Claim]
\label{rem:sfc-scope}
SFC does \textit{not} require that the model ignore user input. It permits
control-authoritative behavior to depend on the semantic content of the request. For example,
the model may refuse harmful queries while complying with benign ones. What it prohibits is
control-authoritative behavior depending on semantically irrelevant representational variation.
The distinction is between \textit{what} the user asks (semantic content, which may legitimately
affect refusal) and \textit{how} the request is encoded (representational form, which should not).
\end{remark}

\paragraph{Operational form.} Let $C$ denote the set of \emph{control-relevant nodes}: positions
and layers whose hidden states causally influence $\Pi_{\mathrm{ctrl}} Y(x)$. In a standard
causal decoder, these include the generation-position hidden states at layers that feed into the
output logits; in agentic systems, they additionally include states feeding into tool-invocation
heads, routing logits, and memory-write authorization heads. For each $c \in C$ at layer $\ell$,
the attention value aggregation is~\citep{vaswani2017attention}:
\begin{equation}
    \widetilde{h}_c^{(\ell)}(x)
    =
    \sum_{j=1}^{|x|} \alpha_{cj}^{(\ell)} W_V^{(\ell)} h_j^{(\ell-1)}(x)
    =
    \underbrace{\sum_{j \in T} \alpha_{cj}^{(\ell)} W_V^{(\ell)} h_j^{(\ell-1)}}_{\text{trusted contribution}}
    +
    \underbrace{\sum_{j \in U} \alpha_{cj}^{(\ell)} W_V^{(\ell)} h_j^{(\ell-1)}}_{\text{untrusted contribution}}
    \label{eq:attn-split}
\end{equation}
where $T$ and $U$ are the trusted and untrusted position sets.

Let $\Gamma_c : \mathbb{R}^d \to \mathbb{R}^{d_c}$ denote the \textbf{control readout map} at
node $c$: the composed function from the hidden state $\widetilde{h}_c^{(\ell)}$ through all
subsequent layers to the representation consumed by the control-authoritative output heads. When
we write $\Pi_{\mathrm{ctrl}}$ applied to a hidden-state difference in the bridge proposition
(Section~\ref{sec:bridge}), we mean $\Gamma_c$ (the hidden-state-level projection) rather than
the behavioral projection $\Pi_{\mathrm{ctrl}} : \Delta(\mathcal{A}) \to
\Delta(\mathcal{A}_{\mathrm{ctrl}})$ defined above.

\section{Threat Model and Scope}
\label{sec:scope}

We consider prompted action model architectures that satisfy all three of the following properties.

\paragraph{(1) Shared-space processing.}
Trusted instructions and untrusted content are processed through a common representational
pipeline. All positions are embedded into a shared hidden space and participate in the same
attention-residual computation. No immutable provenance barrier prevents untrusted representations
from entering the same forward computation as trusted representations.

\paragraph{(2) No immutable provenance enforcement.}
The architecture does not provide a non-user-writable enforcement mechanism that guarantees
trusted and untrusted information remain behaviorally separated. Positional encodings, soft segment
markers, and prompt delimiters do not provide such enforcement, since they are either location
signals or forgeable tokens within the same symbolic channel (Section~\ref{sec:positional}).

\paragraph{(3) Control-authoritative outputs from the shared pipeline.}
Observable control-authoritative outputs $\Pi_{\mathrm{ctrl}} Y(x)$, including refusal
decisions, tool authorization, policy routing, and memory-write authorization, are computed from
control-relevant nodes $C$ whose hidden states arise from the shared attention pipeline.

\paragraph{Architectures outside scope.}
The theorem does \textit{not} apply to architectures that enforce behavioral separation by
construction, including:
\begin{enumerate}
    \item \textbf{Disjoint trusted/untrusted encoders}, where trusted instructions and untrusted
    content are processed by separate representation pipelines and only interact through explicitly
    constrained interfaces.
    \item \textbf{Typed attention or hard segment masks}, where attention from untrusted positions
    into control-relevant trusted computation is structurally forbidden or immutably typed.
    \item \textbf{External policy engines}, where tool authorization, policy routing, or memory
    writes are decided by a mechanism external to the transformer and not directly controlled by
    shared hidden states.
\end{enumerate}
These architectures escape the theorem precisely because they enforce the separation property
that shared-space transformers lack.

\paragraph{Attack classes captured by SFC.}
SFC directly models the following real-world attack classes:
\begin{enumerate}
    \item \textbf{Indirect prompt injection $\rightarrow$ refusal bypass.} Untrusted content
    embedded in retrieved text alters refusal behavior while leaving the nominal task semantics
    unchanged.
    \item \textbf{Retrieved content $\rightarrow$ unauthorized tool invocation.} Untrusted
    external text changes tool-authorization behavior, causing the model to invoke actions outside
    the intended policy.
    \item \textbf{Injected content $\rightarrow$ memory or policy write.} Untrusted content
    modifies memory-write authorization or policy-routing behavior, persisting or escalating
    control influence across turns.
    \item \textbf{Universal adversarial suffixes $\rightarrow$ alignment bypass.}
    Semantically vacuous token strings appended to harmful requests flip control-authoritative
    behavior from refusal to compliance~\citep{zou2023universal}, constituting a direct SFC
    violation.
\end{enumerate}

\paragraph{Necessity of SFC.}
For agentic systems with tool use, memory, routing, or policy enforcement, SFC is a minimal
necessary security property. A weaker requirement would permit untrusted content to alter
authorization decisions through encoding manipulation while still claiming the model is
functioning as intended. A stronger requirement, forbidding any influence of untrusted inputs on
model behavior, would rule out ordinary question answering. SFC permits semantic dependence
through the authorized content channel, but forbids untrusted control over privileged behavior
through representational manipulation.

\section{The Inseparability Theorem}
\label{sec:theorem}

\subsection{Genericity Assumptions}

The following mild assumptions are used by the bridge proposition (Section~\ref{sec:bridge}) to
connect representation-level violations to behavioral change. They are not required for the main
impossibility result.

\begin{assumption}[Generic Control-Path Expressivity]
\label{ass:generic-wv}
For the control-relevant nodes under consideration, the composed linear maps from untrusted value
perturbations into the control-relevant hidden state are nondegenerate. In particular, the induced
map from admissible untrusted perturbations to the control-projected perturbation space is not
identically zero and is not specially aligned so as to annihilate the entire adversarial
perturbation family.
\end{assumption}

\begin{assumption}[Local Control Readout Nondegeneracy]
\label{ass:local-readout}
At the control-relevant states considered in the analysis, the downstream map from the
control-relevant hidden representation to the control-authoritative output distribution is locally
nonconstant on the admissible perturbation directions. Equivalently, the Jacobian of the composed
readout map has nonzero action on those perturbation directions almost everywhere in the region
of interest.
\end{assumption}

\begin{assumption}[Contextual Representation Non-Disjointness]
\label{ass:non-disjoint}
The conditional distributions of contextual representations at control-relevant nodes,
$P_{Z \mid R = \texttt{trusted}}$ and $P_{Z \mid R = \texttt{untrusted}}$, are not disjoint.
That is, the shared attention-residual computation does not map trusted and untrusted inputs to
completely separated regions of activation space.
\end{assumption}

\begin{remark}
Assumptions~\ref{ass:generic-wv} and~\ref{ass:local-readout} are standard genericity conditions.
They rule out pathological parameter settings in which the learned value projections or downstream
nonlinear layers collapse all relevant untrusted perturbations into an exact null space. The results
do not require random weights, only that the learned maps be nondegenerate on the perturbation
family of interest. Assumption~\ref{ass:non-disjoint} states that the shared pipeline does not
accidentally produce perfectly separated representations for trusted and untrusted content. Such
a condition would render the model unable to process user input in the same representational
space as instructions, contradicting the shared-embedding axiom in practice. All three
assumptions are consistent with the empirical success of universal adversarial
suffixes~\citep{zou2023universal}, which demonstrate that adversarial untrusted inputs produce
representations that influence control-relevant computation in all major aligned models tested
to date, suggesting these assumptions are not pathological in practice.
\end{remark}

\subsection{Representational Collision}

\begin{lemma}[Representational Collision]
\label{lem:collision}
Let $E: \Sigma \to \mathbb{R}^d$ be a shared embedding function. Let
$\Sigma_I \subseteq \Sigma$ be the set of tokens appearing in typical instructions, and
$\Sigma_D \subseteq \Sigma$ be the set of tokens appearing in typical user data. Then the
embedding sets overlap:
\begin{equation}
    \{E(t) : t \in \Sigma_I\} \cap \{E(t) : t \in \Sigma_D\} \neq \emptyset.
\end{equation}
\end{lemma}

\begin{proof}
Natural language instructions and data share a common vocabulary. Tokens such as ``the,'' ``is,''
``you,'' ``should,'' ``write,'' ``help,'' and ``answer'' appear in both system instructions and
user data. Since $E$ maps each token to exactly one vector, the same token has the same embedding
regardless of whether it appears in $s$ or $u$. Therefore the embedding sets overlap on at least
$\{E(t) : t \in \Sigma_I \cap \Sigma_D\}$, which is nonempty.
\end{proof}

\subsection{Representational Non-Faithfulness}

\begin{lemma}[Representational Non-Faithfulness]
\label{lem:non-faithful}
For any nontrivial semantic equivalence $\equiv_{\mathrm{sem}}$ and any nonconstant shared
embedding function $E$ (i.e., $E$ is not constant on $\Sigma$), there exist untrusted inputs
$u \equiv_{\mathrm{sem}} u'$ with $u \neq u'$ as token sequences, such that the induced
sequence embeddings satisfy $E(u) \neq E(u')$ (i.e., they differ in at least one position of the
embedding sequence). That is, the shared sequence embedding map does not respect semantic
equivalence: semantically equivalent inputs produce distinct embedding sequences.
\end{lemma}

\begin{proof}
By the nontriviality of $\equiv_{\mathrm{sem}}$, there exist distinct token sequences
$u \neq u'$ with $u \equiv_{\mathrm{sem}} u'$. Since $E$ is nonconstant, there exist tokens
$t \neq t'$ with $E(t) \neq E(t')$. The adversary can construct semantically equivalent
inputs that differ at such token positions (e.g., by appending semantically vacuous suffixes
containing different tokens). The sequence embeddings then differ in at least one
component (or in length). Concrete examples: a request in plaintext vs.\ base64
encoding, a request with vs.\ without an appended semantically vacuous
suffix~\citep{zou2023universal}, or a request in English vs.\ a semantics-preserving translation.
\end{proof}

\subsection{Provenance-Recovery Impossibility}

\begin{theorem}[Provenance-Recovery Impossibility]
\label{thm:prov-imposs}
Let $R \in \{\texttt{trusted}, \texttt{untrusted}\}$ denote the provenance label of a
representation instance $Z$ under test, induced by the per-position provenance variables
$R_j$ of Definition~\ref{def:prompted-input}. Let $Z$ be either a raw token embedding $E(t)$
or a contextual hidden state $h_j^{(\ell)}$ at a control-relevant node. Then the
Bayes-optimal provenance-recovery error is
\begin{equation}
    P_e^\star
    =
    \frac{1}{2}\Bigl(1 - \mathrm{TV}(P_{Z \mid R=\texttt{trusted}},\,
    P_{Z \mid R=\texttt{untrusted}})\Bigr),
    \label{eq:bayes-tv}
\end{equation}
under equal priors, where $\mathrm{TV}$ denotes total variation distance. In particular, perfect
provenance recovery ($P_e^\star = 0$) is impossible unless the two conditional representation
distributions are disjoint.
\end{theorem}

\begin{proof}
Equation~\eqref{eq:bayes-tv} is the standard characterization of Bayes-optimal error for binary
hypothesis testing in terms of total variation
distance~\citep{tsybakov2009introduction}. It remains to show that
$\mathrm{TV}(P_{Z \mid R=\texttt{trusted}}, P_{Z \mid R=\texttt{untrusted}}) < 1$.

When $Z$ is the raw token embedding $E(t)$, Lemma~\ref{lem:collision} gives this directly:
shared vocabulary tokens produce identical embeddings regardless of provenance, so the supports
overlap and TV $< 1$.

For contextual representations $Z = h_j^{(\ell)}$ at control-relevant nodes, we invoke
Assumption~\ref{ass:non-disjoint}.
In either case, $\mathrm{TV} < 1$ and $P_e^\star > 0$.
\end{proof}

\subsection{Control-Path Exposure}

\begin{theorem}[Control-Path Exposure]
\label{thm:control-path}
Let $C$ denote the set of control-relevant nodes of a prompted action model. Suppose that for
some $c \in C$ and layer $\ell$, the attention value aggregation satisfies
\begin{equation}
    \sum_{j \in U} \alpha_{cj}^{(\ell)} > 0,
    \label{eq:control-path-positive}
\end{equation}
where $U$ is the set of untrusted positions. Then untrusted content enters the computation of a
node that causally influences $\Pi_{\mathrm{ctrl}} Y(x)$.
\end{theorem}

\begin{proof}
Equation~\eqref{eq:attn-split} decomposes the control-relevant update into trusted and untrusted
contributions. Condition~\eqref{eq:control-path-positive} implies that the untrusted term is
present with nonzero total attention mass. In a standard causal decoder, the generation-position
hidden states (which are the control-relevant nodes) attend to both the trusted system prefix
and the untrusted user content, so condition~\eqref{eq:control-path-positive} is satisfied
by construction for any useful model. An architecture where $\sum_{j \in U} \alpha_{cj} = 0$ for
all control-relevant nodes would be unable to condition its output on user input, rendering it
functionally useless.
\end{proof}

\subsection{Finite-Coverage Invariance Gap}

\begin{lemma}[Finite-Coverage Invariance Gap]
\label{lem:finite-coverage}
Let $s \in \Sigma^*$ be a fixed trusted prefix and let
\[
    [u]_{\mathrm{sem}} := \{u' \in \Sigma^* : u' \equiv_{\mathrm{sem}} u\}
\]
denote the semantic-equivalence class of an untrusted input $u$. Assume that for some
application-relevant semantics, $[u]_{\mathrm{sem}}$ is infinite or grows super-polynomially
under admissible encoding transformations (paraphrase, language transfer, formatting variation,
character-level obfuscation, Unicode variation, adversarial suffix appending, or other
semantics-preserving rewritings).

Let $\mathcal{T} \subset [u]_{\mathrm{sem}}$ be any finite set of encodings seen during
training, alignment, or adversarial hardening. If the architecture provides no immutable mechanism
enforcing Equation~\eqref{eq:sfc} by construction, then invariance of
$\Pi_{\mathrm{ctrl}} Y(s \| u')$ over all $u' \in [u]_{\mathrm{sem}}$ cannot be guaranteed
solely from correctness on $\mathcal{T}$.
\end{lemma}

\begin{proof}
We argue by contradiction. Suppose, for the sake of contradiction, that correctness of
$\Pi_{\mathrm{ctrl}} Y(s \| \cdot)$ on $\mathcal{T}$ implies invariance of
$\Pi_{\mathrm{ctrl}} Y(s \| u')$ over all $u' \in [u]_{\mathrm{sem}}$. We exhibit an
$u^\star \in [u]_{\mathrm{sem}} \setminus \mathcal{T}$ on which invariance is not certified.

Since $[u]_{\mathrm{sem}}$ is infinite (or grows super-polynomially under admissible encoding
transformations) and $\mathcal{T}$ is finite, the set complement
$[u]_{\mathrm{sem}} \setminus \mathcal{T}$ is nonempty. Pick any
$u^\star \in [u]_{\mathrm{sem}} \setminus \mathcal{T}$. By Lemma~\ref{lem:non-faithful},
since $\equiv_{\mathrm{sem}}$ is nontrivial and $E$ is nonconstant, $u^\star$ produces an
embedding sequence $E(u^\star)$ that is, in general, distinct from the embedding sequences
of all elements in $\mathcal{T}$ at the token level.

The model's forward pass $f \circ E$ is a deterministic function of its input embedding
sequence. Correctness constraints imposed on the finite set
$\{f(E(s \| u')) : u' \in \mathcal{T}\}$ pin down the model's output only on those finitely
many points. Because the architecture provides no immutable mechanism enforcing
Equation~\eqref{eq:sfc} by construction, there is no architectural identification of
$E(u^\star)$ with any element of $\{E(u') : u' \in \mathcal{T}\}$, and consequently no
architectural propagation of the correctness constraint from $\mathcal{T}$ to $u^\star$.

Therefore, correctness on $\mathcal{T}$ does not entail
$\Pi_{\mathrm{ctrl}} Y(s \| u^\star) = \Pi_{\mathrm{ctrl}} Y(s \| u)$, and invariance over
the full class $[u]_{\mathrm{sem}}$ cannot be certified solely from correctness on $\mathcal{T}$.

The gradient-based universal adversarial suffixes of \citet{zou2023universal} provide empirical
confirmation: they construct semantically negligible encoding variations that reliably induce
control-authoritative behavioral changes on aligned models, demonstrating that finite alignment
coverage leaves exploitable invariance gaps in the encoding space.
\end{proof}

\subsection{Main Result: Impossibility of Perfect Semantic-Faithful Control}

\begin{theorem}[Impossibility of Perfect Semantic-Faithful Control]
\label{thm:sfc-imposs}
Assume:
\begin{enumerate}
    \item The model has an exposed control path as in Theorem~\ref{thm:control-path}: untrusted
    content enters control-relevant computation.
    \item Provenance is recovered only from shared representations, with Bayes-optimal error
    given by Theorem~\ref{thm:prov-imposs}.
    \item Application-relevant semantic-equivalence classes are infinite or combinatorially vast,
    and invariance is enforced only through finite training, alignment, or adversarial hardening
    rather than by immutable architectural separation.
\end{enumerate}
Then no mechanism operating solely within the shared representational pipeline can guarantee
perfect Semantic-Faithful Control for all prompted inputs.
\end{theorem}

\begin{proof}
We establish the impossibility by showing that the three premises jointly preclude any
in-pipeline mechanism guaranteeing Equation~\eqref{eq:sfc} for all prompted inputs.

\textit{Step 1: In-pipeline defenses reduce to learned discrimination.}
By Theorem~\ref{thm:control-path}, untrusted content enters control-relevant computation
through the shared value-aggregation pathway. By Theorem~\ref{thm:prov-imposs}, the
Bayes-optimal provenance-recovery error from shared representations is strictly positive
(under Lemma~\ref{lem:collision} for token embeddings, and Assumption~\ref{ass:non-disjoint}
for contextual states), so no perfect provenance-based filter exists within the shared space.
Any defense mechanism operating solely within the shared representational pipeline, lacking
immutable architectural separation, must therefore rely on learned decision boundaries or
learned invariances inside the same shared space, since by hypothesis (premise 3) no other
enforcement mechanism is available to it.

\textit{Step 2: Learned invariances cannot cover the equivalence class.}
Fix any semantic-equivalence class $[u]_{\mathrm{sem}}$ satisfying premise 3, i.e., $[u]_{\mathrm{sem}}$
is infinite or grows super-polynomially under admissible encoding transformations. Let
$\mathcal{T} \subset [u]_{\mathrm{sem}}$ be the (necessarily finite) set of representatives
seen during training, alignment, or adversarial hardening. By
Lemma~\ref{lem:finite-coverage}, correctness on $\mathcal{T}$ does not certify invariance
over all of $[u]_{\mathrm{sem}}$: there exists $u^\star \in [u]_{\mathrm{sem}} \setminus \mathcal{T}$
on which the learned invariance is not constrained.

\textit{Step 3: Combining the steps.}
For this $u^\star$, the model's control-authoritative output $\Pi_{\mathrm{ctrl}} Y(s \| u^\star)$
is determined by a forward pass that includes contributions from $u^\star$'s embeddings via
the exposed control path of Step 1. Since the in-pipeline defense relies on learned decisions
within the shared space (Step 1), and since learned coverage cannot extend to $u^\star$
(Step 2), the defense cannot guarantee
$\Pi_{\mathrm{ctrl}} Y(s \| u^\star) = \Pi_{\mathrm{ctrl}} Y(s \| u)$ for the $u \in \mathcal{T}$
that is semantically equivalent to $u^\star$. Hence Equation~\eqref{eq:sfc} fails on the pair
$(u, u^\star)$, and perfect Semantic-Faithful Control is impossible within the shared
representational pipeline.
\end{proof}

\begin{corollary}[No Perfect In-Pipeline Prevention]
\label{cor:no-perfect}
In any shared-embedding architecture with exposed control paths and no immutable
provenance-enforcing mechanism, perfect prompt-injection prevention is impossible within the
shared representational pipeline.
\end{corollary}

\subsection{Bridge: Representation-Level Violations Propagate Behaviorally}
\label{sec:bridge}

\begin{proposition}[Bridge from Representation-Level Violation to Behavioral Change]
\label{prop:bridge}
Assume Assumptions~\ref{ass:generic-wv} and~\ref{ass:local-readout}. If a semantically
irrelevant untrusted perturbation induces a nonzero control-relevant perturbation at some node
$c \in C$, then the resulting control-authoritative output distribution changes on a set of
perturbations of full measure within the admissible perturbation family.
\end{proposition}

\begin{proof}
Let $\delta z_c = \Gamma_c(\widetilde{h}_c^{(\ell)}(s \| u)) -
\Gamma_c(\widetilde{h}_c^{(\ell)}(s \| u'))$ denote the control-readout perturbation induced by
replacing $u$ with a semantically equivalent $u'$. The full composition from hidden state to
behavioral output is $\Pi_{\mathrm{ctrl}} Y = g \circ \Gamma_c$, where
$g : \mathbb{R}^{d_c} \to \Delta(\mathcal{A}_{\mathrm{ctrl}})$ is the final output map
(e.g., softmax over control logits).

Under Assumption~\ref{ass:generic-wv}, $\delta z_c \neq 0$ for the perturbation family of
interest. Under Assumption~\ref{ass:local-readout}, $g$ is locally nonconstant on the
perturbation directions, so its Jacobian $J_g$ has nonzero action on these directions almost
everywhere in the region of interest. The kernel of $J_g$ is a linear subspace of strictly
lower dimension than the ambient perturbation space, hence a proper submanifold of measure
zero with respect to Lebesgue measure. Therefore the set of perturbation directions that
$g$ annihilates has Lebesgue measure zero. With probability one under any perturbation
distribution absolutely continuous with respect to Lebesgue measure, a nonzero
representation-level violation $\delta z_c$ produces a nonzero behavioral change in
$g(\Gamma_c(\widetilde{h}_c^{(\ell)}(s \| u)))$ versus $g(\Gamma_c(\widetilde{h}_c^{(\ell)}(s \| u')))$.
\end{proof}

\subsection{Quantitative Bound}

\begin{theorem}[Sequence-Level Adversarial Bound]
\label{thm:quantitative}
For $n$-token untrusted sequences, let $Z_n$ denote the induced sequence-level representation.
Let $\mathfrak{A}(B, n)$ denote the class of adversaries with computational budget $B$ and
sequence length $n$. Define:
\begin{equation}
    \delta_n^\star(B)
    :=
    \inf_{\mathcal{D} \in \mathfrak{A}(B,n)}
    \mathrm{TV}\!\left(
    P_{Z_n \mid R = \texttt{trusted}},\,
    P_{Z_n \mid U \sim \mathcal{D}}
    \right).
    \label{eq:adversarial-tv}
\end{equation}
Then the Bayes-optimal error for distinguishing trusted from adversarially chosen untrusted
sequence representations is:
\begin{equation}
    P_{e,n}^\star(B) \geq \frac{1}{2}(1 - \delta_n^\star(B)).
    \label{eq:seq-bound}
\end{equation}
\end{theorem}

\begin{remark}
Intuitively, as the adversary's budget grows or the semantic-equivalence class becomes richer,
$\delta_n^\star(B)$ is expected to decrease, driving the error toward $1/2$ (random guessing).
The formal bound in Equation~\eqref{eq:seq-bound} holds for any value of $\delta_n^\star(B)$;
the monotonic decrease is an empirical expectation rather than a proved property of the bound.
\end{remark}

\begin{proof}
Equation~\eqref{eq:seq-bound} is immediate from the standard TV--Bayes-error identity applied
at the sequence level. The adversary selects the untrusted distribution $\mathcal{D}$ to minimize
TV, exploiting the shared embedding structure to produce representations close to the trusted
distribution. By Assumption~\ref{ass:non-disjoint}, the contextual representation distributions
for trusted and untrusted content are not disjoint, so the infimum in
Equation~\eqref{eq:adversarial-tv} is strictly less than $1$ whenever the adversary class
$\mathfrak{A}(B,n)$ contains a nontrivial distribution over admissible untrusted inputs, and
the error bound follows.
\end{proof}

\section{Positional and Segment Encoding Insufficiency}
\label{sec:positional}

A natural objection is that positional encodings, segment embeddings, or special tokens already
provide partial separation. We address this with a three-tier analysis.

\begin{proposition}[Insufficiency of Positional and Segment Markers]
\label{prop:positional}
Positional and soft segment markers are insufficient for enforcing Semantic-Faithful Control.
\end{proposition}

\paragraph{Tier 1: Positional encodings encode location, not provenance.} Positional encodings
(absolute, relative, or rotary) encode where a token appears in the sequence, not whether it
originates from a trusted or untrusted source. A model can infer approximate provenance from
position (e.g., tokens in the first $|s|$ positions are likely trusted), but this inference is
not an immutable guarantee. Even perfect positional provenance does not imply
non-interference: as Theorem~\ref{thm:control-path} shows, untrusted content enters
control-relevant computation at generation positions regardless of positional labeling.

\paragraph{Tier 2: Soft segment markers are forgeable.} Special tokens such as
\texttt{[INST]}, \texttt{[/INST]}, \texttt{<<SYS>>}, and XML delimiters are elements of
$\Sigma$. They are part of the shared vocabulary and are processed by the shared embedding
function $E$. An adversary can include these markers in untrusted input, producing embeddings
identical to those of legitimate markers. This is not a hypothetical attack: ChatML injection
and similar delimiter-forgery attacks are among the most common prompt injection
techniques~\citep{greshake2023indirect}.

\paragraph{Tier 3: Immutable provenance channels escape the theorem.} If an architecture
augments each embedding with an immutable, non-user-writable provenance tag that is preserved
through all layers and cannot be modified by the attention mechanism, this constitutes the
architectural separation advocated by our result. Such a mechanism is precisely the class of
solution our theorem implies is necessary. Early instances of this approach exist:
ASIDE~\citep{zverev2025aside} gives instruction and data tokens distinct embeddings through a
transformation applied from the first layer, providing exactly the non-forgeable,
architecture-level separation our result requires.

\section{The Von Neumann Isomorphism}
\label{sec:vonneumann}

We formalize the structural analogy between prompt injection and buffer overflow. This analogy
is explanatory and motivational; it is not part of the formal proof.

\begin{center}
\begin{tabular}{@{}ll@{}}
\toprule
\textbf{Von Neumann Architecture} & \textbf{Transformer Architecture} \\
\midrule
Shared memory (code + data) & Shared embedding space (instructions + data) \\
Program counter & Attention mechanism \\
Machine instruction & System prompt token \\
User data in memory & User input embedding \\
Buffer overflow exploit & Prompt injection attack \\
Stack canary / ASLR & Guardrail classifier / perplexity filter \\
NX bit / W$\oplus$X (hardware) & Architectural separation (proposed) \\
Memory-safe languages & ``Semantically safe'' representational paradigms (open) \\
\bottomrule
\end{tabular}
\end{center}

The risks of code and data confusion in shared memory were flagged
early~\citep{anderson1972planning}, and buffer overflow exploitation was characterized over the
following two decades. No single mitigation fully solved the problem. Progress came from a
layered combination of defenses: bounds checking in source code and safer library functions (the
root-cause fix), stack canaries (introduced by StackGuard in 1998), address space layout
randomization (ASLR), non-executable memory pages (NX bit, introduced by AMD in 2003 with
widespread OS support in 2004), control-flow integrity, and ultimately the adoption of
memory-safe languages (Rust, Go, Java) that eliminate the bug class at the language level. Each
mechanism closed specific exploitation patterns---NX bit was bypassed by return-to-libc and ROP
attacks, canaries by precise memory disclosures, ASLR by information leaks---but the combination
raised the cost of exploitation substantially. The total time from discovery to mature
defense-in-depth: approximately three decades, with memory-safe languages continuing to displace
unsafe code to this day.

Prompt injection was identified circa 2022. Our theorem provides the formal justification for
why software-only defenses within the shared representational pipeline will continue to be
circumvented. The analogy suggests that a mature defense posture for prompt injection will
likely require a comparable layered approach: architectural separation of instruction and data
channels (analogous to NX bit), runtime mitigations (analogous to ASLR and stack canaries),
and eventually representational paradigms that eliminate the shared-embedding confusion at the
architectural level (analogous to memory-safe languages). The inseparability theorem identifies
the root cause; the full defense posture is a long-horizon engineering program.

\section{Empirical Grounding}
\label{sec:empirical}

While the inseparability theorem is a formal result that does not require empirical validation,
we provide three measurements that ground the theoretical quantities in real systems.

\subsection{Vocabulary Overlap (M1)}

We measured the vocabulary overlap between instruction corpora and user data corpora across three
major tokenizers. The instruction corpus was compiled from four public system prompt
repositories (TheBigPromptLibrary, jujumilk3/leaked-system-prompts, YeeKal/leaked-system-prompts,
asgeirtj/system\_prompts\_leaks). The user data corpus comprised 50{,}000 user turns streamed
from UltraChat-200K~\citep{ding2023ultrachat}. Tokenizers were loaded via HuggingFace
\texttt{AutoTokenizer} with special tokens excluded and inputs truncated at 2{,}048 tokens.
Metrics were computed over full unigram frequency distributions. We report three metrics: type
overlap $\rho_{\mathrm{type}} = |V_I \cap V_D| / |V_I \cup V_D|$ (Jaccard index),
probability-mass overlap $\rho_{\mathrm{mass}} = \sum_t \min\{p_I(t), p_D(t)\}$, and
instruction coverage $\rho_{\mathrm{cov}} = |V_I \cap V_D| / |V_I|$ (fraction of instruction
tokens also appearing in user data).

\begin{table}[ht]
\centering
\caption{Vocabulary overlap between instruction and user data corpora across three tokenizers.
$\rho_{\mathrm{cov}}$ is the metric most relevant to Lemma~\ref{lem:collision}: it measures what
fraction of the instruction vocabulary is available to an adversary.}
\label{tab:vocab-overlap}
\begin{tabular}{@{}lrrccc@{}}
\toprule
\textbf{Tokenizer} & $|V_I|$ & $|V_D|$ & $\rho_{\mathrm{type}}$ & $\rho_{\mathrm{mass}}$ & $\rho_{\mathrm{cov}}$ \\
\midrule
Mistral-7B (32K)   & 23{,}898 & 22{,}476 & 0.736 & 0.547 & 0.823 \\
Llama-3.1 (128K)   & 52{,}238 & 50{,}378 & 0.502 & 0.517 & 0.657 \\
Qwen-2.5 (152K)    & 54{,}104 & 49{,}249 & 0.474 & 0.519 & 0.615 \\
\bottomrule
\end{tabular}
\end{table}

Across all three tokenizers, a substantial majority of instruction-vocabulary tokens also appear
in user data ($\rho_{\mathrm{cov}}$ ranges from 61.5\% to 82.3\%). This confirms
Lemma~\ref{lem:collision}: the embedding sets for instructions and data overlap extensively,
not merely on a few tokens but across the majority of the instruction vocabulary. The lower
Jaccard overlap for larger-vocabulary tokenizers (Llama 3.1, Qwen 2.5) reflects vocabulary
dilution from specialized tokens (CJK characters, code tokens) that appear in neither corpus,
not a reduction in the adversary's effective attack surface.

\subsection{Representation-Level Overlap (M2)}

To ground Assumption~\ref{ass:non-disjoint}, we measured the Maximum Mean Discrepancy
(MMD) between hidden-state distributions for instruction-origin and user-origin sequences across
layers of Llama 3.1 8B Instruct (8-bit quantized). We randomly sampled $N = 300$ sequences per
corpus, extracted mean-pooled hidden states at layers 0, 8, 16, 24, and 31, and computed the
unbiased MMD$^2$ estimator with an RBF kernel using the median bandwidth heuristic. Bootstrap
95\% confidence intervals were computed with $B = 1{,}000$ resamples.

\begin{table}[ht]
\centering
\caption{MMD values by layer with bootstrap 95\% confidence intervals.}
\label{tab:mmd}
\begin{tabular}{@{}ccc@{}}
\toprule
\textbf{Layer} & \textbf{MMD} & \textbf{95\% CI} \\
\midrule
0  & 0.606 & $[0.572,\; 0.640]$ \\
8  & 0.908 & $[0.860,\; 0.953]$ \\
16 & 0.885 & $[0.846,\; 0.925]$ \\
24 & 0.805 & $[0.772,\; 0.839]$ \\
31 & 0.755 & $[0.725,\; 0.789]$ \\
\bottomrule
\end{tabular}
\end{table}

\begin{figure}[ht]
\centering
\includegraphics[width=0.75\textwidth]{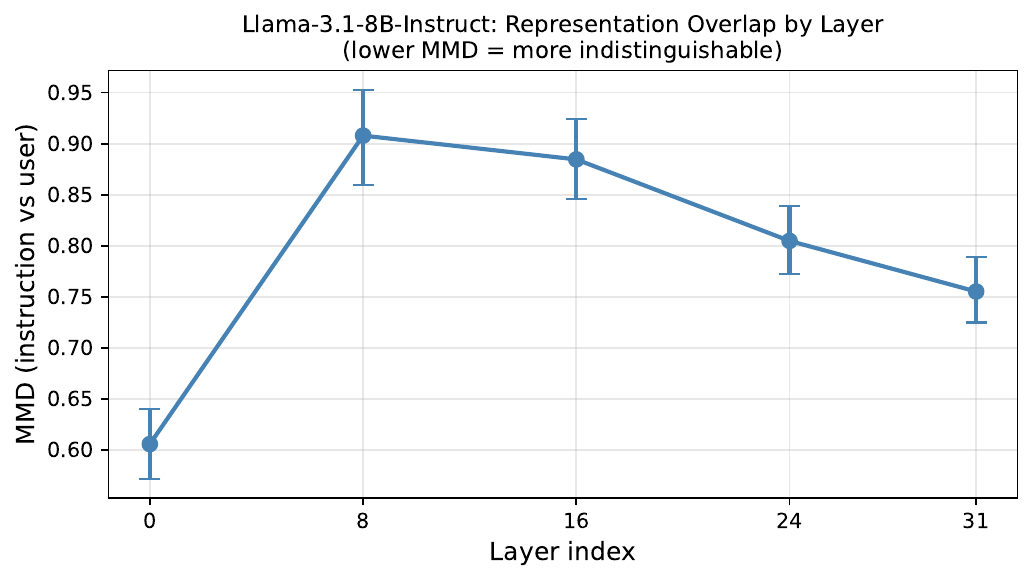}
\caption{MMD between instruction-origin and user-origin hidden-state distributions by layer in
Llama 3.1 8B Instruct. Lower MMD indicates more indistinguishable distributions. Error bars show
bootstrap 95\% confidence intervals ($B = 1{,}000$). The distributions retain substantial
overlap at every layer, consistent with Assumption~\ref{ass:non-disjoint}.}
\label{fig:mmd}
\end{figure}

The results (Figure~\ref{fig:mmd}, Table~\ref{tab:mmd}) show that at layer 0, MMD $\approx
0.61$, meaning instruction and user representations are relatively indistinguishable at the raw
embedding level. MMD peaks at layer 8 ($\approx 0.91$), indicating that the model develops
partial internal separation in early-middle layers. This separation then monotonically declines
through layers 16, 24, and 31 (MMD $\approx 0.76$), the layers most relevant to
control-authoritative output computation. At no layer do the distributions become fully
separated. This pattern is consistent with the inseparability thesis: the shared
attention-residual computation partially differentiates instruction and user content in middle
layers but cannot sustain this separation through to the control-relevant output layers.

We note that recent efficiency-driven advances in KV-cache and representation quantization, such
as TurboQuant~\citep{zandieh2026turboquant}, compress hidden states to as few as 3 bits with
near-zero accuracy loss. Because quantization is a deterministic many-to-one map, it can only
preserve or reduce the total variation distance between trusted and untrusted representation
distributions (by the data-processing inequality), never increase it. The industry trend toward
aggressive shared-space compression therefore tightens, rather than relaxes, the assumptions
underlying Assumption~\ref{ass:non-disjoint}.

\subsection{Behavioral SFC Violation (M3)}

To connect the geometric measurements to observable behavior, we tested whether untrusted
content alters control-authoritative outputs in practice. We constructed a financial services
system prompt with strict refusal rules and a baseline trusted investment query. We then prepended
the same trusted query to five prompt injection techniques: DAN-style role override, many-shot
priming, fake admin system override, Crescendo-style escalation, and virtualization/persona
hijack. All experiments used greedy decoding with a maximum of 350 new tokens.

We tested two models: Llama 3.1 8B Instruct (heavily RLHF-hardened) and Mistral 7B Instruct v0.2
(less hardened). Both models refused the harmful payloads on all five attack techniques. Neither
model complied with the injected instructions.

However, an SFC violation does not require compliance. Both models produced structurally
different outputs compared to the baseline on multiple attacks: altered response length, added
cautionary preambles, changed framing, and modified refusal language, all driven entirely by
injected tokens that added no legitimate semantic content to the trusted query.

We note an important scoping caveat. Our formal $\Pi_{\mathrm{ctrl}}$ projection covers
refusal decisions, tool authorization, policy routing, and memory-write authorization. The
output-level shifts observed here (length, preamble, framing) are not direct $\Pi_{\mathrm{ctrl}}$
violations under the strict reading of Definition~\ref{def:sfc}, since the refusal decision
itself remains constant. They are evidence of \emph{control-path exposure} in the sense of
Theorem~\ref{thm:control-path}: untrusted tokens demonstrably reach and modify control-relevant
hidden states in ways that propagate to observable output. They do not, on their own, establish
a direct refusal flip in this experiment. Direct refusal flips are documented elsewhere by GCG
and abliteration (Section~\ref{sec:bridge}); the M3 measurements complement those by showing
that even when refusal does not flip, the control path is provably contaminated.

This finding is significant for two reasons. First, it demonstrates that untrusted content
\textit{does} influence control-relevant computation even in heavily hardened models, confirming
the control-path exposure of Theorem~\ref{thm:control-path}. Second, Llama 3.1 resists known
attack signatures while showing identical M1 and M2 geometry to Mistral. This confirms that RLHF
operates as behavioral patching over an architectural vulnerability. The underlying
representational inseparability (measured by M1 and M2) is identical across both models; only
the learned behavioral response differs. This is precisely the distinction the inseparability
theorem draws: architectural properties cannot be resolved by training-time interventions.

\subsection{Weight-Surgery Confirmation: Abliteration (M4)}

A particularly striking confirmation of the inseparability thesis comes from a class of attacks
known as \textit{abliteration}~\citep{labonne2024abliteration, arditi2024refusal}. Building on
the finding that refusal behavior in aligned LLMs is mediated by a single direction in the
residual stream, abliteration identifies this direction by contrasting model activations on
harmful and harmless prompts, then projects it out of the model's weight matrices. The
reported result is a model with substantially reduced refusal capability, achieved without
retraining or input-side adversarial optimization.

The technique has been reported as effective on multiple aligned model families including Llama,
Mistral, Gemma, and Qwen, with many abliterated variants publicly released on community model
repositories. Refinements such as projected abliteration~\citep{lai2025projected} report
reduced collateral performance impact, suggesting that the safety direction can be isolated
with limited disturbance to the model's other capabilities.

We treat abliteration as empirical motivation rather than as part of our formal evidentiary
backbone, since much of the reporting appears in community blog posts and non-archival venues.
With that scoping, abliteration aligns with the inseparability thesis in three ways:

\textbf{Consistency with control-path exposure (Theorem~\ref{thm:control-path}).} If a low-rank
projection on model weights substantially reduces refusal behavior across inputs, this is
consistent with safety being encoded as a direction in the shared representational space rather
than enforced by an isolated, architecturally separated safety circuit. The refusal mechanism
appears to share the residual stream with other model behavior, which is what the control-path
exposure theorem requires.

\textbf{Consistency with finite-coverage limitations (Lemma~\ref{lem:finite-coverage}).}
Defenses against abliteration, such as extended-refusal
fine-tuning~\citep{shairah2025extended}, attempt to disperse the safety signal across multiple
latent dimensions through training-time intervention. The reported pattern is that such defenses
raise the difficulty of attack but do not eliminate the underlying vulnerability. This is
consistent with Lemma~\ref{lem:finite-coverage}: finite hardening can increase the adversary's
required effort but cannot certify invariance over all semantically irrelevant representational
variations.

\textbf{Consistency with the architectural-vs-training distinction.} Abliteration applies the
same projection technique across model families and training regimes, regardless of how much
RLHF or constitutional AI training the base model received. This is consistent with current
safety alignment operating within the shared representational pipeline, adding learned decision
boundaries rather than enforcing structural separation, which is the only kind of intervention
the Inseparability Theorem predicts is available within shared-embedding architectures.

The pattern of attack and defense around abliteration follows the broader pattern predicted by
the Von Neumann analogy: each new defense raises the adversary's cost but operates within the
same shared representational space, rather than providing architectural separation. We view
this as empirical motivation for the theorem's architectural conclusion, with the formal
evidentiary weight resting on the proof itself, the M1/M2 measurements, and the published GCG
construction.

\section{Connection to Existing Results}
\label{sec:connections}

\paragraph{Representation-level safety.} Recent work has shown that safety alignment in LLMs
is mediated by low-dimensional structures in activation space: \citet{arditi2024refusal}
demonstrated that refusal behavior is controlled by a single direction, removable via
orthogonal projection. \citet{wei2023jailbroken} analyzed the mechanisms by which safety training
fails, identifying competing objectives and mismatched generalization as root causes. These
findings are consistent with our result: if safety is encoded as a fragile, low-dimensional
feature in a shared representational space, it cannot be robustly separated from adversarial
user input that operates in the same space.

\paragraph{Adversarial robustness of aligned models.} \citet{carlini2023aligned} showed that
aligned models remain vulnerable to adversarial attacks, demonstrating a fundamental gap between
alignment training and adversarial robustness. Our theorem provides a structural explanation for
this gap: alignment training can improve average-case behavior but cannot guarantee worst-case
robustness within a shared-embedding architecture.

\paragraph{Impossibility results in AI.} \citet{brcic2023impossibility} surveyed impossibility
theorems relevant to AI, categorizing them into deduction, indistinguishability, induction,
tradeoffs, and intractability. Our result falls in the \textit{indistinguishability} category and
is, to our knowledge, the first impossibility theorem specific to prompt injection.

\paragraph{Computational indistinguishability.} Our use of ``indistinguishability'' is related
to but distinct from the cryptographic notion~\citep{goldwasser1984probabilistic}. In
cryptography, computational indistinguishability means no polynomial-time algorithm can distinguish
two distributions. In our setting, the indistinguishability is \textit{representational}: semantically
equivalent inputs of different encodings produce distinct embeddings that cannot be reliably
normalized within the shared pipeline.

\paragraph{Formalization of prompt injection.} \citet{liu2024formalizing} provided the first
formalization of prompt injection attacks and defenses at USENIX Security 2024. Their work
focuses on \textit{measuring} attack success rates. Our result complements theirs by establishing
that no defense achievable within the shared embedding framework can have zero failure rate.

\paragraph{Instruction-data separation.} \citet{zverev2025separation} introduced a formal
measure of instruction-data separation, an empirical variant computable from model outputs, and
the SEP benchmark, finding that no evaluated model achieves high separation and that prompt
engineering and fine-tuning do not substantially close the gap. Their work quantifies the
separation deficit; ours explains it. Where they measure how far real models fall short, we prove
that within shared-embedding architectures the deficit cannot be driven to zero by any
in-pipeline mechanism, so their empirical finding is what our impossibility result predicts.
Their measure is an average-case, output-level quantity, whereas our property is a worst-case
guarantee over the encoding class, which is why our result takes the form of an impossibility
rather than a benchmark, and points to the architectural direction realized by
ASIDE~\citep{zverev2025aside}.

\paragraph{Classical noninterference.} Our security property draws on the noninterference
framework of \citet{goguen1982security}, which requires that actions by one security domain do not
affect observations in another. \citet{sabelfeld2003language} generalized this to language-based
information-flow security, emphasizing that labeling data as untrusted is insufficient without
enforcement that untrusted data cannot affect privileged control flow. Our SFC property is a
specialization of this principle to transformer architectures, where the ``enforcement''
corresponds to architectural separation of control and data channels.

\paragraph{OWASP LLM Top 10.} Prompt injection has been the \#1 vulnerability in the OWASP
Top 10 for Large Language Model Applications since its inception~\citep{owasp2025}, in the
development of which one of us participated. Our theorem provides the formal justification for
why it has proven so resistant to mitigation.

\section{Discussion}
\label{sec:discussion}

\paragraph{What the theorem does and does not claim.}
We do not claim that useful models must ignore user input. We do not claim that all defenses are
equally ineffective. We claim that in shared-embedding architectures, no mechanism operating
solely within the shared representational pipeline can \textit{guarantee} that
control-authoritative behavior depends only on semantic content. This is an impossibility of
\textit{perfect} prevention, not a claim that all mitigation is futile. Defense-in-depth
remains valuable; what is ruled out is provable completeness within the current architecture.

\paragraph{Assumptions and their justification.}
The two genericity assumptions (Assumptions~\ref{ass:generic-wv} and~\ref{ass:local-readout})
rule out pathological weight configurations in which all adversarial perturbations are silently
annihilated. These assumptions are empirically validated by the success of universal adversarial
suffixes~\citep{zou2023universal} and many-shot jailbreaking~\citep{anil2024many} across all
major aligned models. A model that violated these assumptions would be one in which no
representational perturbation could ever alter control behavior. Such a model would be
functionally unresponsive to user input.

\paragraph{Practical impact.}
Organizations deploying LLM-integrated applications should treat prompt injection the way they
treat buffer overflows: assume it \textit{will} happen, implement defense-in-depth (multiple
layers of imperfect mitigations), and limit the blast radius through sandboxing and privilege
separation. No single guardrail should be trusted as a complete solution.

\section{Conclusion}
\label{sec:conclusion}

We have proven that perfect prompt injection prevention is impossible within shared-embedding
architectures that lack enforced control-data separation. The proof combines three independently
necessary results: provenance-recovery impossibility (shared representations prevent reliable
origin classification), control-path exposure (untrusted content enters control-relevant
computation via shared attention), and finite-coverage invariance gap (finite training cannot
certify encoding invariance over infinite semantic-equivalence classes).

This result places prompt injection in the same category as buffer overflows, a vulnerability
class that required decades of layered defenses (bounds checking, stack canaries, ASLR, NX bit,
control-flow integrity) and ultimately the adoption of memory-safe languages to meaningfully
contain. The path forward for prompt injection is analogous: the AI community must invest in
\textit{structural separation} of instruction and data channels at the architectural level,
complemented by runtime mitigations and, in the longer term, representational paradigms that
eliminate the shared-embedding confusion by construction. Until such architectural separation
is achieved, prompt injection will remain an irreducible risk in LLM-integrated applications
that process untrusted content through shared representational pipelines, amenable to
defense-in-depth mitigation but not to elimination.

\bibliographystyle{plainnat}

\end{document}